\documentclass[fleqn,12pt,twoside]{article}
\usepackage{tipa}

\usepackage[headings]{espcrc1}
\readRCS
$Id: espcrc1.tex,v 1.2 2004/02/24 11:22:11 spepping Exp $
\usepackage{graphicx}
\usepackage[figuresright]{rotating}

\newtheorem{theorem}{Theorem}

\newtheorem{case}{Case}

\newtheorem{corollary}{Corollary}

\newtheorem{remark}{Remark}

\newenvironment{proof}[1][Proof.]{\begin{trivlist}
\item[\hskip \labelsep {\bfseries #1}]}{\end{trivlist}}

\newcommand{\AmS}{{\protect\the\textfont2
  A\kern-.1667em\lower.5ex\hbox{M}\kern-.125emS}}

\usepackage{amsmath}
\usepackage{amsfonts}
\title{Estimates for the number of vertices with an interval spectrum in proper edge colorings of some graphs}

\author{R.R. Kamalian\address[MCSD]{Institute for Informatics and Automation Problems, National Academy
of Sciences of RA, 0014 Yerevan, Republic of Armenia}%
\thanks{email: rrkamalian@yahoo.com}}

\runtitle{Estimates for the number of vertices with an
interval spectrum in proper edge colorings of some graphs} \runauthor{R.R. Kamalian}

\begin{document}

\maketitle

\begin{abstract}
A proper edge $t$-coloring of a graph $G$ is a
coloring of edges of $G$ with colors $1,2,...,t$ such that each of
$t$ colors is used, and adjacent edges are colored differently. The
set of colors of edges incident with a vertex $x$ of $G$ is called a
spectrum of $x$. A proper edge $t$-coloring of a graph $G$ is
interval for its vertex $x$ if the spectrum of $x$ is an interval of
integers. A proper edge $t$-coloring of a graph $G$ is
persistent-interval for its vertex $x$ if the spectrum of $x$ is an
interval of integers beginning from the color $1$.

For graphs $G$ from some classes of graphs, we obtain estimates for
the possible number of vertices for which a proper edge $t$-coloring
of $G$ can be interval or persistent-interval.
\bigskip
\end{abstract}

\section{Introduction}

We consider undirected, simple, finite, connected graphs. For a
graph $G$, we denote by $V(G)$ and $E(G)$ the sets of its vertices
and edges, respectively. For any $x\in V(G)$, $d_{G}(x)$ denotes the
degree of the vertex $x$ in $G$. For a graph $G$, we denote by
$\Delta (G)$ the maximum degree of a vertex of $G$. A function
$\varphi:E(G)\rightarrow\{1,2,\ldots,t\}$ is called a proper edge
$t$-coloring of a graph $G$ if each of $t$ colors is used, and
adjacent edges are colored differently. The set of all proper edge
$t$-colorings of $G$ is denoted by $\alpha(G,t).$ The minimum value
of $t$ for which there exists a proper edge $t$-coloring of a graph
$G$ is called a chromatic index \cite{Vizing2} of $G$ and is denoted
by $\chi'(G).$ Let us also define the set $\alpha(G)$ of all proper
edge colorings of the graph $G$
$$\alpha(G)\equiv \bigcup_{t=\chi'(G)}^{|E(G)|}\alpha(G,t).$$

If $G$ is a graph, $\varphi\in\alpha(G)$, $x\in V(G)$, then the set
of colors of edges of $G$ incident with $x$ is called a spectrum of
the vertex $x$ in the coloring $\varphi$ of the graph $G$ and is
denoted by $S_G(x,\varphi).$

An arbitrary nonempty subset of consecutive integers is called an
interval. An interval with the minimum element $p$ and the maximum
element $q$ is denoted by $[p,q]$. An interval $D$ is called a
$h$-interval if $|D|=h$.

For any real number $\xi$, we denote by $\lfloor\xi \rfloor$
($\lceil\xi \rceil$) the maximum (minimum) integer which is less
(greater) than or equal to $\xi$.

If $G$ is a graph, $\varphi\in\alpha(G)$, and $x\in V(G)$, then we
say that $\varphi$ is interval (persistent-interval) for $x$ if
$S_G(x,\varphi)$ is a $d_G(x)$-interval (a $d_G(x)$-interval with
$1$ as its minimum element). For an arbitrary graph $G$ and any
$\varphi\in\alpha(G)$, we denote by
$f_{G,i}(\varphi)(f_{G,pi}(\varphi))$ the number of vertices of the
graph $G$ for which $\varphi$ is interval (persistent-interval). For
any graph $G$, let us \cite{Shved2_} set
$$\eta_i(G)\equiv \max_{\varphi\in\alpha(G)}f_{G,i}(\varphi),\quad
\eta_{pi}(G)\equiv \max_{\varphi\in\alpha(G)}f_{G,pi}(\varphi).$$

A bipartite graph $G$ with bipartition $(X,Y)$ is called
$(a,b)$-biregular, if $d_G(x)=a$ for any vertex $x\in X$, and
$d_G(y)=b$ for any vertex $y\in Y$.

The terms and concepts that we do not define can be found in
\cite{West1}.

It is clear that if for any graph $G$ $\eta_{pi}(G)=|V(G)|$, then
$\chi'(G)=\Delta(G)$. For a regular graph $G$, these two conditions
are equivalent: $\eta_{pi}(G)=|V(G)|$ iff $\chi'(G)=\Delta(G)$. It
is known \cite{Holyer3,Leven} that for a regular graph $G$, the
problem of deciding whether $\chi'(G)=\Delta(G)$ or not is
$NP$-complete. It means that for a regular graph $G$, the problem of
deciding whether $\eta_{pi}(G)=|V(G)|$ or not is also $NP$-complete.
For any tree $G$, some necessary and sufficient condition for
$\eta_{pi}(G)=|V(G)|$ was obtained in \cite{Caro5}. In this paper,
for an arbitrary regular graph $G$, we obtain a lower bound for the
parameter $\eta_{pi}(G)$.

If $G$ is a graph, $R_0\subseteq V(G)$, and the coloring
$\varphi\in\alpha(G)$ is interval (persistent-interval) for any
$x\in R_0$, then we say that $\varphi$ is interval
(persistent-interval) on $R_0$.

$\varphi\in\alpha(G)$ is called an interval coloring of a graph $G$
if $\varphi$ is interval on $V(G)$.

We define the set $\mathfrak{N}$ as the set of all graphs for which
there is an interval coloring. Clearly, for any graph $G$,
$G\in\mathfrak{N}$ if and only if $\eta_i(G)=|V(G)|$.

The notion of an interval coloring was introduced in \cite{Oranj6}.
In \cite{Oranj6,Diss7,Asratian8} it is shown that if
$G\in\mathfrak{N}$, then $\chi'(G)=\Delta(G)$. For a regular graph
$G$, these two conditions are equivalent: $G\in\mathfrak{N}$ iff
$\chi'(G)=\Delta(G)$ \cite{Oranj6,Diss7,Asratian8}. Consequently,
for a regular graph $G$, four conditions are equivalent:
$G\in\mathfrak{N},$ $\chi'(G)=\Delta(G),$ $\eta_i(G)=|V(G)|,$
$\eta_{pi}(G)=|V(G)|$. It means that for any regular graph $G$,
\begin{enumerate}
\item the problem of deciding whether or not $G$ has an interval
coloring is $NP$-complete,
\item the problem of deciding whether
$\eta_i(G)=|V(G)|$ or not is $NP$-complete.
\end{enumerate}

In this paper, for an arbitrary regular graph $G$, we obtain a lower
bound for the parameter $\eta_i(G)$.

We also obtain some results for bipartite graphs. The complexity of
the problem of existence of an interval coloring for bipartite
graphs is investigated in \cite{Asratian9,Giaro,Sev}. In
\cite{Diss7} it is shown that for a bipartite graph $G$ with
bipartition $(X,Y)$ and $\Delta(G)=3$ the problem of existence of a
proper edge $3$-coloring which is persistent-interval on $X\cup Y$
(or even only on $Y$ \cite{Oranj6,Diss7}) is $NP$-complete.

Suppose that $G$ is an arbitrary bipartite graph with bipartition
$(X,Y)$. Then $\eta_i(G)\geq\max\{|X|,|Y|\}$.

Suppose that $G$ is a bipartite graph with bipartition $(X,Y)$ for
which there exists a coloring $\varphi\in\alpha(G)$
persistent-interval on $Y$. Then $\eta_{pi}(G)\geq 1+|Y|$.

Some attention is devoted to $(a,b)$-biregular bipartite graphs
\cite{Asratian10,Hanson11,Hanson12,Kost13} in the case $b=a+1$.

We show that if $G$ is a $(k-1,k)$-biregular bipartite graph, $k\geq
4$, then
$$\eta_i(G)\geq\frac{k-1}{2k-1}\cdot|V(G)|+\Bigg\lceil\frac{k}{\big\lceil\frac{k}{2}\big\rceil\cdot(2k-1)}\cdot|V(G)|\Bigg\rceil.$$

We show that if $G$ is a $(k-1,k)$-biregular bipartite graph, $k\geq
3$, then
$$\eta_{pi}(G)\geq\frac{k}{2k-1}\cdot|V(G)|.$$

\section{Results}

\begin{theorem}\label{Theorem1} \cite{Shved2_}
If $G$ is a regular graph with $\chi'(G)=1+\Delta(G)$, then
$$\eta_{pi}(G)\geq\Bigg\lceil\frac{|V(G)|}{1+\Delta(G)}\Bigg\rceil.$$
\end{theorem}
\begin{proof}
Suppose that $\beta\in\alpha(G,1+\Delta(G))$. For any
$j\in[1,1+\Delta(G)]$, define
$$V_{G,\beta,j}\equiv\{x\in V(G)/ j\not\in S_G(x,\beta)\}.$$

For arbitrary integers $j',j'',$ where $1\leq j'< j''\leq
1+\Delta(G),$ we have
$$V_{G,\beta,j'}\cap V_{G,\beta,j''}=\emptyset$$
and $$ \bigcup_{j=1}^{1+\Delta(G)}V_{G,\beta,j}=V(G).$$

Hence, there exists $j_0\in [1,1+\Delta(G)]$ for which
$$|V_{G,\beta,j_0}|\geq\Bigg\lceil\frac{|V(G)|}{1+\Delta(G)}\Bigg\rceil.$$

Set $R_0\equiv V_{G,\beta,j_0}.$

\case{1}$j_0=1+\Delta(G).$

Clearly, $\beta$ is persistent-interval on $R_0$.

\case{2} $j_0\in[1,\Delta(G)].$

Define a function $\varphi:E(G)\rightarrow [1,1+\Delta(G)]$. For any
$e\in E(G),$ set:
$$
\varphi(e)\equiv\left\{
\begin{array}{ll}
\beta(e), & \textrm{if $\beta(e)\not\in \{j_0,1+\Delta(G)\}$}\\
j_0, & \textrm{if $\beta(e)=1+\Delta(G)$}\\
1+\Delta(G), & \textrm{if $\beta(e)=j_0$.}
\end{array}
\right.
$$

It is not difficult to see that $\varphi\in\alpha(G,1+\Delta(G))$
and $\varphi$ is persistent-interval on $R_0$.

\end{proof}

\begin{corollary}\label{Corollary1} \cite{Shved2_}
If $G$ is a cubic graph, then there exists a coloring from
$\alpha(G,\chi'(G))$ which is persistent-interval for at least
$\Big\lceil\frac{|V(G)|}{4}\Big\rceil$ vertices of $G$.
\end{corollary}

\begin{theorem}\label{Theorem2} \cite{Shved2_}
If $G$ is a regular graph with $\chi'(G)=1+\Delta(G)$, then
$$\eta_i(G)\geq\Bigg\lceil\frac{|V(G)|}{\big\lceil\frac{1+\Delta(G)}{2}\big\rceil}\Bigg\rceil.$$
\end{theorem}
\begin{proof}
Suppose that $\beta\in\alpha(G,1+\Delta(G))$. For any
$j\in[1,1+\Delta(G)]$, define
$$V_{G,\beta,j}\equiv\{x\in V(G)/ j\not\in S_G(x,\beta)\}.$$

For arbitrary integers $j',j'',$ where $1\leq j'< j''\leq
1+\Delta(G),$ we have
$$V_{G,\beta,j'}\cap V_{G,\beta,j''}=\emptyset $$
and
$$ \bigcup_{j=1}^{1+\Delta(G)}V_{G,\beta,j}=V(G).$$

For any $i\in[1,\big\lceil\frac{1+\Delta(G)}{2}\big\rceil]$, let us
define the subset $V(G,\beta,i)$ of the set $V(G)$ as follows:
$$
V(G,\beta,i)\equiv\left\{
\begin{array}{ll}
V_{G,\beta,2i-1}\cup V_{G,\beta,2i}, & \textrm{if $\Delta(G)$ is odd and $i\in[1,\frac{1+\Delta(G)}{2}]$}\\
 & \textrm{or $\Delta(G)$ is even and $i\in[1,\frac{\Delta(G)}{2}]$,}\\
V_{G,\beta,1+\Delta(G)}, & \textrm{if $\Delta(G)$ is even and
$i=1+\frac{\Delta(G)}{2}$.}\\
\end{array}
\right.
$$

For arbitrary integers $i',i'',$ where $1\leq i'< i''\leq
\big\lceil\frac{1+\Delta(G)}{2}\big\rceil,$ we have
$$V(G,\beta,i')\cap V(G,\beta,i'')=\emptyset$$
and
$$\bigcup_{i=1}^{\big\lceil\frac{1+\Delta(G)}{2}\big\rceil}V(G,\beta,i)=V(G).$$

Hence, there exists $i_0\in
\big[1,\big\lceil\frac{1+\Delta(G)}{2}\big\rceil\big]$ for which
$$|V(G,\beta,i_0)|\geq\Bigg\lceil\frac{|V(G)|}{\big\lceil\frac{1+\Delta(G)}{2}\big\rceil}
\Bigg\rceil.$$

Set $R_0\equiv V(G,\beta,i_0).$

\case{1} $i_0=\big\lceil\frac{1+\Delta(G)}{2}\big\rceil.$

\case{1.a} $\Delta(G)$ is even.

Clearly, $\beta$ is interval on $R_0$.

\case{1.b} $\Delta(G)$ is odd.

Define a function $\varphi:E(G)\rightarrow [1,1+\Delta(G)]$. For any
$e\in E(G),$ set:
$$
\varphi(e)\equiv\left\{
\begin{array}{ll}
(\beta(e)+1)(\bmod{(1+\Delta(G))}), & \textrm{if $\beta(e)\neq\Delta(G)$,}\\
1+\Delta(G), & \textrm{if $\beta(e)=\Delta(G)$.}\\
\end{array}
\right.
$$

It is not difficult to see that $\varphi\in\alpha(G,1+\Delta(G))$
and $\varphi$ is interval on $R_0$.

\case{2} $1\leq i_0\leq\big\lceil\frac{\Delta(G)-1}{2}\big\rceil.$

Define a function $\varphi:E(G)\rightarrow [1,1+\Delta(G)]$. For any
$e\in E(G),$ set:
$$
\varphi(e)\equiv\left\{
\begin{array}{ll}
(\beta(e)+2+\Delta(G)-2i_0)(\bmod{(1+\Delta(G))}), & \textrm{if $\beta(e)\neq 2i_0-1$,}\\
1+\Delta(G), & \textrm{if $\beta(e)=2i_0-1$.}\\
\end{array}
\right.
$$

It is not difficult to see that $\varphi\in\alpha(G,1+\Delta(G))$
and $\varphi$ is interval on $R_0$.
\end{proof}

\begin{corollary}\label{Corollary2} \cite{Shved2_}
If $G$ is a cubic graph, then there exists a coloring from
$\alpha(G,\chi'(G))$ which is interval for at least
$\frac{|V(G)|}{2}$ vertices of $G$.
\end{corollary}

\begin{theorem}\label{Theorem3}\cite{Oranj6,Diss7,Asratian8}
Let $G$ be a bipartite graph with bipartition $(X,Y)$. Then there
exists a coloring $\varphi\in\alpha(G,|E(G)|)$ which is interval on
$X$.
\end{theorem}

\begin{corollary}\label{Corollary3}
 Let $G$ be a bipartite graph with bipartition $(X,Y)$. Then
$\eta_i(G)\geq\max\{|X|,|Y|\}$.
\end{corollary}

\begin{theorem} \cite{Asratian14,Oranj6,Asratian8}
Let $G$ be a bipartite graph with bipartition $(X,Y)$ where
$d_G(x)\leq d_G(y)$ for each edge $(x,y)\in E(G)$ with $x\in X$ and
$y\in Y$. Then there exists a coloring
$\varphi_0\in\alpha(G,\Delta(G))$ which is persistent-interval on
$Y$.
\end{theorem}

\begin{theorem}
Suppose $G$ is a bipartite graph with bipartition $(X,Y)$, and there
exists a coloring $\varphi_0\in\alpha(G,\Delta(G))$ which is
persistent-interval on $Y$. Then, for an arbitrary vertex $x_0\in
X$, there exists $\psi\in\alpha(G,\Delta(G))$ which is
persistent-interval on $\{x_0\}\cup Y$.
\end{theorem}

\begin{proof}
\case{1} $S_G(x_0,\varphi_0)=[1,d_G(x_0)]$. In this case $\psi$ is
$\varphi_0$.

\case{2} $S_G(x_0,\varphi_0)\neq [1,d_G(x_0)]$.

Clearly, $[1,d_G(x_0)]\backslash S_G(x_0,\varphi_0)\neq\emptyset$,
$S_G(x_0,\varphi_0)\backslash [1,d_G(x_0)]\neq\emptyset$. Since
$|S_G(x_0,\varphi_0)|=|[1,d_G(x_0)]|=d_G(x_0)$, there exists
$\nu_0\in[1,d_G(x_0)]$ satisfying the condition
$|[1,d_G(x_0)]\backslash
S_G(x_0,\varphi_0)|=|S_G(x_0,\varphi_0)\backslash
[1,d_G(x_0)]|=\nu_0$.

Now let us construct the sequence
$\Theta_0,\Theta_1,\ldots,\Theta_{\nu_0}$ of proper edge
$\Delta(G)$-colorings of the graph $G$, where for any $i\in
[0,\nu_0]$, $\Theta_i$ is persistent-interval on $Y$.

Set $\Theta_0\equiv\varphi_0$.

Suppose that for some $k\in [0,\nu_0-1]$, the subsequence
$\Theta_0,\Theta_1,\ldots,\Theta_k$ is already constructed.

Let
$$
t_k\equiv\max(S_G(x_0,\Theta_k)\backslash [1,d_G(x_0)]),
$$
$$
s_k\equiv\min([1,d_G(x_0)]\backslash S_G(x_0,\Theta_k)).
$$

Clearly, $t_k>s_k$. Consider the path $P(k)$ in the graph $G$ of
maximum length with the initial vertex $x_0$ whose edges are
alternatively colored by the colors $t_k$ and $s_k$. Let
$\Theta_{k+1}$ is obtained from $\Theta_k$ by interchanging the two
colors $t_k$ and $s_k$ along $P(k)$.

It is not difficult to see that $\Theta_{\nu_0}$ is
persistent-interval on $\{x_0\}\cup Y$. Set
$\psi\equiv\Theta_{\nu_0}$.
\end{proof}

\begin{corollary}
Let $G$ be a bipartite graph with bipartition $(X,Y)$ where
$d_G(x)\leq d_G(y)$ for each edge $(x,y)\in E(G)$ with $x\in X$ and
$y\in Y$. Let $x_0$ be an arbitrary vertex of $X$. Then there exists
a coloring $\varphi_0\in\alpha(G,\Delta(G))$ which is
persistent-interval on $\{x_0\}\cup Y$.
\end{corollary}

\begin{corollary} \label{Cor5} \cite{Shved2_}
Let $G$ be a bipartite graph with bipartition $(X,Y)$ where
$d_G(x)\leq d_G(y)$ for each edge $(x,y)\in E(G)$ with $x\in X$ and
$y\in Y$. Then $\eta_{pi}(G)\geq 1+|Y|$.
\end{corollary}

\begin{remark}\label{Remark1}
Notice that the complete bipartite graph $K_{n+1,n}$ for an
arbitrary positive integer $n$ satisfies the conditions of Corollary
\ref{Cor5}. Is is not difficult to see that
$\eta_{pi}(K_{n+1,n})=1+n.$ It means that the bound obtained in
Corollary \ref{Cor5} is sharp since in this case $|Y|=n.$
\end{remark}

\begin{remark}\label{Remark2}
Let $G$ be a bipartite $(k-1,k)$-biregular graph with bipartition
$(X,Y)$, where $k\geq3$. Then the numbers $\frac{|X|}{k}$,
$\frac{|Y|}{k-1}$, and $\frac{|V(G)|}{2k-1}$ are integer. It follows
from the equalities $gcd(k-1,k)=1$ and
$|E(G)|=|X|\cdot(k-1)=|Y|\cdot k$.
\end{remark}

\begin{theorem}\label{Theorem5} \cite{Shved2_}
Let $G$ be a bipartite $(k-1,k)$-biregular graph, where $k\geq 4$.
Then
$$\eta_i(G)\geq\frac{k-1}{2k-1}\cdot|V(G)|+\Bigg\lceil\frac{k}{\big\lceil\frac{k}{2}\big\rceil\cdot(2k-1)}\cdot|V(G)|\Bigg\rceil.$$
\end{theorem}
\begin{proof}
Suppose that $(X,Y)$ is a bipartition of $G$. Clearly,
$\chi'(G)=\Delta(G)=k$. Suppose that $\beta\in\alpha(G,k)$. For any
$j\in[1,k]$, define:
$$V_{G,\beta,j}\equiv\{x\in X/j\not\in S_G(x,\beta)\}.$$

For arbitrary integers $j',j'',$ where $1\leq j'< j''\leq k,$ we
have
$$V_{G,\beta,j'}\cap V_{G,\beta,j''}=\emptyset$$
and

$$ \bigcup_{j=1}^{k}V_{G,\beta,j}=X.$$

For any $i\in[1,\lceil\frac{k}{2}\rceil]$, let us define the subset
$V(G,\beta,i)$ of the set $X$ as follows:
$$
V(G,\beta,i)\equiv\left\{
\begin{array}{ll}
V_{G,\beta,2i-1}\cup V_{G,\beta,2i}, & \textrm{if $k$ is odd and $i\in[1,\frac{k-1}{2}]$}\\
 & \textrm{or $k$ is even and $i\in[1,\frac{k}{2}]$,}\\
V_{G,\beta,k}, & \textrm{if $k$ is odd and
$i=\frac{1+k}{2}$.}\\
\end{array}
\right.
$$

For arbitrary integers $i',i'',$ where $1\leq i'< i''\leq
\big\lceil\frac{k}{2}\big\rceil,$ we have
$$V(G,\beta,i')\cap V(G,\beta,i'')=\emptyset$$
and
$$
\bigcup_{i=1}^{\big\lceil\frac{k}{2}\big\rceil}V(G,\beta,i)=X.$$

Hence, there exists $i_0\in
\big[1,\big\lceil\frac{k}{2}\big\rceil\big]$ for which
$$|V(G,\beta,i_0)|\geq\Bigg\lceil\frac{|X|}{\big\lceil\frac{k}{2}\big\rceil}
\Bigg\rceil.$$

Set $R_0\equiv Y\cup V(G,\beta,i_0).$

It is not difficult to verify that
$$|R_0|\geq\frac{k-1}{2k-1}\cdot|V(G)|+\Bigg\lceil\frac{k}{\big\lceil\frac{k}{2}\big\rceil\cdot(2k-1)}\cdot|V(G)|\Bigg\rceil.$$

\case{1} $i_0=\big\lceil\frac{k}{2}\big\rceil$.

\case{1.a} $k$ is odd.

Clearly, $\beta$ is interval on $R_0$.

\case{1.b} $k$ is even.

Define a function $\varphi:E(G)\rightarrow [1,k]$. For any $e\in
E(G)$, set:
$$
\varphi(e)\equiv\left\{
\begin{array}{ll}
(\beta(e)+1)(\bmod{k}), & \textrm{if $\beta(e)\neq k-1$,}\\
k, & \textrm{if $\beta(e)=k-1$.}\\
\end{array}
\right.
$$

It is not difficult to see that $\varphi\in\alpha(G,k)$ and
$\varphi$ is interval on $R_0$.

\case{2} $i_0\in\big[1,\big\lceil\frac{k}{2}\big\rceil-1\big].$

Define a function $\varphi:E(G)\rightarrow [1,k]$. For any $e\in
E(G)$, set:
$$
\varphi(e)\equiv\left\{
\begin{array}{ll}
(\beta(e)+1+k-2i_0)(\bmod{k}), & \textrm{if $\beta(e)\neq 2i_0-1$,}\\
k, & \textrm{if $\beta(e)=2i_0-1$.}\\
\end{array}
\right.
$$

It is not difficult to see that $\varphi\in\alpha(G,k)$ and
$\varphi$ is interval on $R_0$.
\end{proof}

\begin{corollary}\label{Corollary4} \cite{Shved2_}
Let $G$ be a bipartite $(k-1,k)$-biregular graph, where $k$ is even
and $k\geq 4$. Then
$$\eta_i(G)\geq\frac{k+1}{2k-1}\cdot|V(G)|.$$
\end{corollary}

\begin{corollary}\label{Corollary55} \cite{Shved2_}
Let $G$ be a bipartite $(3,4)$-biregular graph. Then there exists a
coloring from $\alpha(G,4)$ which is interval for at least
$\frac{5}{7}|V(G)|$ vertices of $G$.
\end{corollary}

\begin{remark}
For an arbitrary bipartite graph $G$ with $\Delta(G)\leq 3$, there
exists an interval coloring of $G$ \cite{Hansen15,Giaro16,Giaro17}. Consequently, if $G$ is a bipartite $(2,3)$-biregular
graph, then $\eta_i(G)=|V(G)|$.
\end{remark}

\begin{remark}
Some sufficient conditions for existence of an interval coloring of
a $(3,4)$-biregular bipartite graph were obtained in
\cite{Asratian18,Asratian19,Piatkin20}.
\end{remark}

\begin{theorem}\label{Theorem6} \cite{Shved2_}
Let $G$ be a bipartite $(k-1,k)$-biregular graph, where $k\geq 3$.
Then
$$\eta_{pi}(G)\geq\frac{k}{2k-1}\cdot|V(G)|.$$
\end{theorem}
\begin{proof}
Suppose that $(X,Y)$ is a bipartition of $G$. Clearly,
$\chi'(G)=\Delta(G)=k$. Suppose that $\beta\in\alpha(G,k)$.

For any $j\in[1,k]$, define:
$$V_{G,\beta,j}\equiv\{x\in X/j\not\in S_G(x,\beta)\}.$$

For arbitrary integers $j',j'',$ where $1\leq j'< j''\leq k,$ we
have
$$V_{G,\beta,j'}\cap V_{G,\beta,j''}=\emptyset$$
and
$$ \bigcup_{j=1}^{k}V_{G,\beta,j}=X.$$

Hence, there exists $j_0\in [1,k]$ for which
$$|V_{G,\beta,j_0}|\geq\frac{|X|}{k}.$$

Set $R_0\equiv Y\cup V_{G,\beta,j_0}.$

It is not difficult to verify that
$$|R_0|\geq\frac{k}{2k-1}\cdot|V(G)|.$$

\case{1} $j_0=k$.

Clearly, $\beta$ is persistent-interval on $R_0$.

\case{2} $j_0\in[1,k-1]$.

Define a function $\varphi:E(G)\rightarrow [1,k]$. For any $e\in
E(G)$, set:
$$
\varphi(e)\equiv\left\{
\begin{array}{ll}
\beta(e), & \textrm{if $\beta(e)\not\in \{j_0,k\}$}\\
j_0, & \textrm{if $\beta(e)=k$}\\
k, & \textrm{if $\beta(e)=j_0$.}
\end{array}
\right.
$$

It is not difficult to see that $\varphi\in\alpha(G,k)$ and
$\varphi$ is persistent-interval on $R_0$.
\end{proof}

\begin{corollary}\label{Corollary5} \cite{Shved2_}
Let $G$ be a bipartite $(3,4)$-biregular graph. Then there exists a
coloring from $\alpha(G,4)$ which is persistent-interval for at
least $\frac{4}{7}|V(G)|$ vertices of $G$.
\end{corollary}

\section{Acknowledgment}

The author thanks professors A.S. Asratian and P.A. Petrosyan for
their attention to this work.

\end{document}